\documentclass[a4paper,11pt]{article}
\usepackage[margin=1in]{geometry}
\usepackage{amsmath,amssymb,amsthm,textcomp,stmaryrd}
\usepackage[inline]{enumitem}
\usepackage[binary-units]{siunitx}
\usepackage[sort&compress,square,sort,comma,numbers]{natbib}
\usepackage[mathscr]{euscript}
\usepackage[T1]{fontenc}
\usepackage[utf8]{inputenc}
\usepackage[final]{microtype}
\usepackage[hyphens]{url}
\usepackage[hidelinks]{hyperref}

\newtheorem{theorem}{Theorem}[section]
\newtheorem{lemma}[theorem]{Lemma}
\newtheorem{proposition}[theorem]{Proposition}

\theoremstyle{definition}
\newtheorem{invariant}[theorem]{Invariant}
\newtheorem{definition}[theorem]{Definition}
\theoremstyle{remark}
\newtheorem{example}[theorem]{Example}

\newcommand{\Service}{\mathfrak{S}}
\newcommand{\Clients}{\mathfrak{C}}
\newcommand{\Replicas}{\mathfrak{R}}
\newcommand{\Faulty}{\mathscr{F}}

\newcommand{\Replica}[1][r]{\MakeUppercase{#1}}
\newcommand{\Client}[1][c]{\MakeLowercase{#1}}
\newcommand{\ID}[1]{\mathop{\textsf{id}}(#1)}
\newcommand{\NonFaulty}[1]{\mathop{\textsf{nf}}(#1)}
\newcommand{\n}{\textbf{n}}
\newcommand{\f}{\textbf{f}}

\newcommand{\m}{\textbf{m}}
\newcommand{\Instances}{\mathscr{I}}
\newcommand{\Instance}[1]{\mathcal{I}_{#1}}
\newcommand{\Primary}[1]{\mathcal{P}_{#1}}

\newcommand{\rn}{\rho}
\newcommand{\Request}{\texttt{CR}}
\newcommand{\Decs}[1]{\mathscr{D}_{#1}}
\newcommand{\Dec}[2]{\mathcal{D}_{#1,#2}}
\newcommand{\Succ}[1]{\textsf{S}(#1)}
\newcommand{\Fail}{\textsf{F}}
\newcommand{\SuccS}[1]{\textsf{S}(#1)}
\newcommand{\FailS}[1]{\textsf{F}(#1)}

\newcommand{\Hash}[1]{\mathop{\texttt{Hash}}(#1)}

\newcommand{\abs}[1]{\lvert #1 \rvert}
\newcommand{\union}{\cup}

\newcommand{\difference}{\setminus}
\newcommand{\Concat}{\oplus}
\newcommand{\subref}[2]{\ref{#1}\ref{#2}}

\renewcommand{\div}{\operatorname{div}}
\renewcommand{\bmod}{\operatorname{mod}}

\usepackage[noend]{algorithmic}

\newcommand{\GETS}{:=}
\newcommand{\Var}[1]{\textit{#1\/}}
\newenvironment{myprotocol}{
    \hrule
    \smallskip
    \scriptsize
    \algsetup{linenosize=\tiny}
    \begin{algorithmic}[1]

        \newcommand{\PROTOCOL}[2]{\item[] \textbf{\underline{##1}} ##2\\[2pt]}
        \makeatletter
            \newcommand{\EVENT}[1]{\STATE \textbf{event} ##1 \textbf{:}\begin{ALC@g}}
            \newcommand{\ENDEVENT}{\end{ALC@g}}
        \makeatother
}{
    \end{algorithmic}
    \smallskip
    \hrule
}

\newcommand{\pBFT}{\textsc{Pbft}}
\newcommand{\rBFT}{\textsc{Rbft}}
\newcommand{\BFT}{\textsc{bft}}

\usepackage{tikz}
\usetikzlibrary{arrows.meta,decorations.pathmorphing,decorations.pathreplacing}
\tikzset{
    >=Stealth,
    dot/.style={circle,scale=0.35,draw=black,fill=black},
    node_text/.append style={font=\strut\bfseries},
    label/.append style={font=\strut\footnotesize},
    decosnake/.append style={decoration={snake,pre length=4pt,post length=6pt,segment length=4,amplitude=.9},decorate}
}

\begin{document}

\title{Revisiting consensus protocols through wait-free parallelization\footnote{A brief announcement of this work will be presented at the 33rd International Symposium on Distributed Computing (DISC 2019)~\cite{mbftba}.}}
\author{Suyash Gupta \and Jelle Hellings \and Mohammad Sadoghi}
\date{\normalsize{\begin{tabular}{c}
                            Exploratory Systems Lab\\
                            Department of Computer Science\\
                            University of California, Davis\\
                            CA 95616-8562, USA
                \end{tabular}}}

\maketitle

\begin{abstract}
The recent surge of blockchain systems has renewed the interest in traditional Byzantine fault-tolerant consensus protocols. Many such consensus protocols have a \emph{primary-backup} design in which an assigned replica, \emph{the primary}, is responsible for coordinating the consensus protocol. Although the primary-backup design leads to relatively simple and high-performance consensus protocols, it places an \emph{unreasonable burden} on a good primary and allows malicious primaries to substantially \emph{affect} the system performance.

In this paper, we propose a protocol-agnostic approach to improve the design of primary-backup consensus protocols. At the core of our approach is a novel wait-free approach of running several instances of the underlying consensus protocol \emph{in parallel}. To yield a high-performance parallelized design, we present coordination-free techniques to order operations across parallel instances, deal with instance failures, and assign clients to specific instances. Consequently, the design we present is able to reduce the load on individual instances and primaries, while also reducing the adverse effects of any malicious replicas.
\end{abstract}

\section{Introduction}\label{sec:intro}

The introduction of \emph{Bitcoin}---the first wide-spread application driven by \emph{blockchains}---has resulted in a surge of interest in blockchain technology. This interest is backed by many use cases in the public and private sectors. 
For example, in trade~\cite{pwcenergy,christies,impactblock,pwcenergy,hypereal}, identity management~\cite{blockdev,hypereal,impactblock}, food production~\cite{wurblockfood}, aid delivery~\cite{blockdev,hypereal}, health care~\cite{blockhealthover,blockhealthfac,blockeu}, fraud prevention~\cite{promiseblock}, and GDPR compliance~\cite{ibmgdpr}. 
At the core of these use cases is the need to manage and replicate data, such as financial transactions, among a group of participants. Consequently, at the core of blockchain technology are \emph{consensus protocols} that allow replicating data across a group of servers (replicas), some of which can fail or can act maliciously.

Several use cases for blockchains operate in a \emph{permissioned environment} in which the participants can only join via well-established procedures.
These established procedures prevent malicious entities from controlling a majority of the replicas. 
In this permissioned environment, a blockchain can be maintained using traditional high-performance Byzantine fault-tolerant (\BFT{}) consensus protocols~\cite{pbft,pbftj,aardvark,rbft,zyzzyva,zyzzyvaj,bft700,byza,fastbft,linbft,sbft,blockplane,trans-book,clussendba}. Commonly, these protocols use the \emph{primary-backup model} pioneered in the Practical Byzantine Fault Tolerant consensus protocol~\cite{pbft}. 
In these \BFT{} protocols, a single replica is designated as \emph{the primary} and is responsible for coordinating the consensus decisions, 
while all the other replicas perform the \emph{backup role}.

The primary-backup model simplifies the development of consensus protocols substantially: when a primary is non-malicious, then even the simplest broadcast replication protocols suffice. The only complication in these consensus protocols is the way in which they deal with malicious primaries: malicious behavior must be either detected (after which the primary can be replaced) or prevented altogether. 
This simplicity of the primary-backup model negatively affects its performance in three ways~\cite{rbft,aardvark,prime,spin}:
\begin{enumerate}[wide]
\item \emph{Primary load}. The primary not only has to perform the primary tasks, but also the backup role (as it is itself a replica). Consequently, the primary receives a higher load than other replicas, 
and this load at the primary can become a \emph{bottleneck in the overall system throughput}. This is especially the case in fine-tuned high-performance consensus protocols that employ complex cryptographic primitive, for example, to reduce communication overheads or to improve resilience.
\item \emph{Primary replacement}. As stated earlier, primary-backup consensus protocols work only when the primary behaves in accordance with the protocol. If the primary acts malicious or is faulty, then it will be replaced. However, detection of such behaviors requires setting timers. Further, replacing a faulty primary usually takes a while. During this time the system is unable to handle requests, which {\em negatively affects its overall throughput}.
\item \emph{Malicious behavior}. Primary-backup consensus protocols rely on the \emph{underlying algorithm} to detect malicious behavior of the primary. Usually, these detectors are only capable of detecting catastrophic failures that prevent new consensus decisions altogether, but they fail to detect or deal with primaries that {\em affect the performance of the system in other ways}, for example, a malicious primary could reduce or throttle the throughput of the system.
\end{enumerate}

To the best of our knowledge, no approach is yet able to address all these limitations of primary-backup consensus protocols. In this work, we address these limitations in a \emph{protocol-agnostic} manner by exploiting parallelization. 
In our \emph{paradigm}, we run several instances of the underlying consensus protocol \emph{in parallel} and we balance the system load among these parallel instances. 
This parallelism helps to reduce the load per primary and mitigates the negative impacts of a single primary on the throughput of the system. 
Our design is fine-tuned such that the instances coordinated by non-faulty replicas are \emph{wait-free}: they can continuously make consensus decisions, independent of the behavior of any other instances. Our paradigm is highly flexible: it can be used in combination with any well-behaved  primary-backup consensus protocol and it can be fine-tuned towards various application-specific needs.

\paragraph{Organization.}
In Section~\ref{sec:prelim}, we introduce terminology and notations that we use throughout this paper. In Section~\ref{sec:stepwise}, we present how we envision parallelization of consensus protocols, and tackle the main design challenges. Next, in Section~\ref{sec:waitfree}, we refine the heavily step-wise coordinated approach of Section~\ref{sec:stepwise} by presenting a \emph{wait-free}  design that adds additional challenges but supports maintaining high throughput. Finally, in Section~\ref{sec:related}, we discuss related work and in Section~\ref{sec:conclude} we conclude on our findings and discuss avenues for further research.

\section{Preliminaries}\label{sec:prelim}

In this work, we present a protocol-agnostic paradigm to parallelize consensus protocols with the aim of increasing consensus throughput, while reducing the effects of individual malicious replicas. We now introduce the notations and assumptions used throughout this paper.

\paragraph{Service notation.}
We represent a replicated \emph{service} by a triple $\Service = (\Clients, \Replicas, \Faulty)$, where $\Clients$ is the set of clients using the service, $\Replicas$ is the set of \emph{replicas} and $\Faulty \subset \Replicas$ is the set of \emph{faulty replicas} that exhibit Byzantine behavior. We write $\n = \abs{\Replicas}$ and $\f = \abs{\Faulty}$ to denote the number of replicas and faulty replicas, respectively. We assign each replica $\Replica \in \Replicas$ a unique identifier $\ID{\Replica}$ with $0 \leq \ID{\Replica} < \n$. Similarly, we assign each client $\Client \in \Clients$ a unique identifier $\ID{\Client}$ with $0 \leq \ID{\Client} < \abs{\Clients}$. The set of \emph{non-faulty replicas}, denoted by $\NonFaulty{\Service}$, is defined as $\NonFaulty{\Service} = \Replicas \difference \Faulty$. We assume that the non-faulty replicas behave in accordance with the protocol and are deterministic: on identical inputs, we expect non-faulty replicas to produce identical outputs.

\paragraph{Consensus protocol.}
A \emph{consensus protocol} helps to replicate a sequence of \emph{values} among all the non-faulty replicas. A single execution of a correct consensus protocol satisfies the following two requirements~\cite{distalgo}:
\begin{enumerate}
    \item \emph{Termination}. Each non-faulty replica accepts a value.
    \item \emph{Non-divergence}. All non-faulty replicas accept the same value.
\end{enumerate}
In this paper, we consider those consensus protocols that replicate a 
sequence of client requests (for example, database operations). 
In this setting, the termination and non-divergence requirements imply \emph{data consistency}, a safety property. 
Additionally, termination of one \emph{round} (or one consensus) assures that all the non-faulty replicas 
have the same state.
This ensures that any preconditions for the next round are met, which implies \emph{availability}, a liveness property. 
Most general-purpose consensus protocols that do not expect synchronous communication to guarantee 
non-divergence require $\n > 3\f$, which we also assume throughout this paper~\cite{pbft,pbftj}.

The focus of this paper are the consensus protocols that follow the \emph{primary-backup model}. 
In such protocols, a single replica is assigned the role of the \emph{primary} and  is responsible for initiating and 
coordinating each round of the consensus protocols, while all the remaining replicas perform the \emph{backup} role.  
As the primary can be malicious, these protocols usually have the means to \emph{detect failure} of the primary and \emph{transfer control} to a new primary. A well-known example of a primary-backup consensus protocol is the  Practical Byzantine Fault Tolerance protocol (\pBFT{})~\cite{pbft,pbftj}, which has inspired the design of many modern consensus protocols~\cite{aardvark,zyzzyva,zyzzyvaj,bft700,byza,rbft,fastbft}.

\paragraph{Sequence and map notations.}
Let $S = [s_0, \dots, s_{k-1}]$ be a sequence. We write $S[i]$ to denote $s_i$ and $\abs{S}$ to denote the length $k$ of $S$. If $v$ is a value, then $S \Concat v = [s_0, \dots, s_{k-1}, v]$ denotes the \emph{concatenation} of $S$ and $v$. If $v$ is a value, then $S \difference v$ denotes the sequence obtained from $S$ by removing all occurrences of $v$. If $k$ and $v$ are values, then we write $k \mapsto v$ to denote a \emph{key-value} mapping that maps $k$ onto $v$.

\paragraph{Cryptographic primitives.}
We assume a \emph{collision-resistant hash function} that maps an arbitrary value $v$ to a numeric value  $\Hash{v}$ in a bounded range, called the \emph{digest}~\cite{hac}. We assume that it is practically impossible to find another value $v'$, $v \neq v'$, such that $\Hash{v} = \Hash{v'}$.

\section{Parallelizing consensus}\label{sec:stepwise}

We now present our paradigm to parallelize a consensus protocol. 
In our paradigm, each replica $\Replica \in \Replicas$ participates in $\m$, $1 \leq \m \leq \n$, instances 
of the underlying consensus protocol. 
We use $\Instances = \{ \Instance{1}, \dots, \Instance{\m} \}$ to denote these instances, 
where $1, \dots, \m$ acts as the \emph{instance identifier}. 
We write $\Instance{i}(\Replica)$ to denote the $i$-th instance, $1 \leq i \leq \m$, running on a replica $\Replica$ and
 we use $\Instances(\Replica) = \{ \Instance{1}(\Replica), \dots, \Instance{\m}(\Replica) \}$ to represent 
all the $\m$ instances running, in parallel, on the replica $\Replica$. 
Further, we represent the primary of an $i$-th instance as $\Primary{i}$, $1 \leq i \leq \m$. 
Our paradigm enforces that each primary exists on a distinct replicas, that is, for all $1 \leq i < j \leq \m$, we have $\Primary{i} \neq \Primary{j}$.

\begin{figure}[t]
    \centering
    \begin{tikzpicture}[scale=0.9,transform shape]
        \draw[draw=black,rounded corners=5pt,fill=black!10!blue!10] (0, 0) rectangle (5, 2);
        \draw[draw=black,rounded corners=5pt,fill=black!20!blue!20] (6, 0) rectangle (11, 2);
        \draw[draw=black,rounded corners=5pt,fill=black!30!blue!30] (12, 0) rectangle (17, 2);
        
        \node[align=center] at (2.5, 1) {Run a round of $\m$ \BFT{}\\ instances to accept $\m$ client\\ requests in parallel.};
        \node[align=center] at (8.5, 1) {Collect accepted requests\\and create an order among\\ the requests.};
        \node[align=center] at (14.5, 1) {Execute the requests in\\the created order};

        \node[above,font=\bfseries,align=center] at (2.5, 2) {Parallelized Consensus};
        \node[above,font=\bfseries,align=center] at (8.5, 2) {Ordering};
        \node[above,font=\bfseries,align=center] at (14.5, 2) {Execution};
        
        \path[thick] (5,1) edge[->] (6,1) (11, 1) edge[->] (12,1);
    \end{tikzpicture}
    \caption{A basic flow of tasks undertaken by a replica while employing our paradigm. Each replica is given the same \BFT{} protocol.}
    \label{fig:overview}
\end{figure}
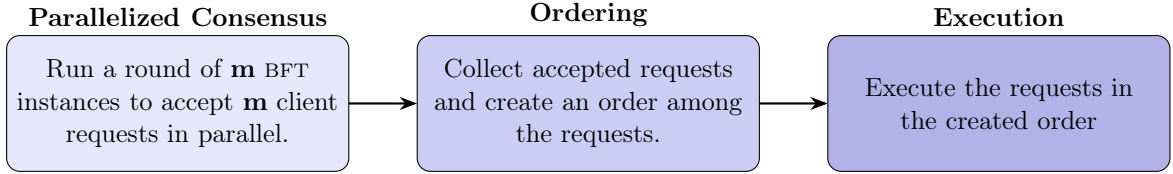

Figure~\ref{fig:overview} presents a set of tasks undertaken by a replica employing our paradigm. Each replica takes as input a \BFT{} protocol and runs $\m$ instance of that protocol in parallel. Once these instances complete, the replica generates a global order of the requests across all these instances and executes these requests in the global order. For the sake of clarity, we first present a parallelized design in which the instances operate in a coordinated \emph{step-wise} manner. Furthermore, we assume that each instance operates a general consensus protocol: 

\begin{definition}\label{def:bbcon}
We model a \emph{consensus protocol} as a black-box that operates in well-defined \emph{rounds}. 
In each such round, a single \emph{consensus decision} is made by \emph{all} the non-faulty replicas. 
If a round \emph{succeeds}, then $\Succ{\Request}$ is the consensus decision observed by \emph{all} the non-faulty replicas, 
where $\Request$ is the client request accepted by all the non-faulty replicas in that round. 
If a round \emph{fails}, then $\Fail$ is the consensus decision observed by \emph{all} the non-faulty replicas, which indicates a primary failure. 
After a failure, a replica can be instructed to \emph{transfer control} to a new primary.  
If all the non-faulty replicas are instructed to transfer control to the same new primary, then 
this process will succeed and a new primary is elected.%
\footnote{Several practical \BFT{}-style consensus protocols do not strictly adhere to these assumptions. To improve throughput, these protocols provide \emph{partial consensus}, in which a majority of the non-faulty replicas are guaranteed to make successful consensus decisions. In these partial consensus protocols, consensus among all the non-faulty replicas is guaranteed only eventually through additional checkpoint and recovery steps.%
%In Appendix~\ref{app:practical_path}, we discuss how one can extend our design to such partial consensus protocols.
}
\end{definition}

The \emph{general} model of Definition~\ref{def:bbcon} allows us to focus on the core challenges in parallelizing consensus protocols. At the core of our paradigm is the coordination of $\m$ instances of a consensus protocol running \emph{in parallel}. This implies that a single round of our paradigm coordinates \emph{multiple parallel consensus rounds}, each of which is initiated and managed by a \emph{distinct} primary $\Primary{i}$ for the instance $\Instance{i}$, $1\leq i \leq \m$. Each consensus decision succeeds whenever $\Primary{i}$ is non-faulty. 
This approach to parallelization raises several important challenges:
\begin{enumerate}
    \item For optimal throughput, we need to ensure that each instance is making a distinct consensus decision, that is, 
each instance is processing a distinct client request.
    \item Every non-faulty replica should execute all the accepted client requests in \emph{the same order}.
    \item When several instances fail in a round and want to transfer control to new primaries, then all non-faulty replicas need to do so in the same manner.
\end{enumerate}
In our design, we address each of these challenges. Figure~\ref{fig:mbft_replica_overview} sketches a high-level overview of a \emph{parallelized consensus round} at replica $\Replica$.

\begin{figure*}
    \centering
    \begin{tikzpicture}[xscale=0.975,yscale=0.5]
        \node[dot] (sor) at (0.5, 0) {};
        \node[node_text,left] at (sor) {Start};

        \path (1,  4) edge (1,  -3)
              (3,  4) edge (3,  -3)
              (6,  4) edge (6,  -3)
              (14, 4) edge (14, -3);
        \node[label,above,align=center] at (2  , 2) {Parallel\\consensus};
        \node[label,above,align=center] at (4.5, 2) {Collect\\decisions};

        \node (i1) at (2,  1.5) {$\Instance{1}(\Replica)$};
        \node (i2) at (2,  0.5) {$\Instance{2}(\Replica)$};
        \node      at (2, -0.25) {$\vdots$};
        \node (im) at (2, -1.5) {$\Instance{\m}(\Replica)$};

        \path[->,semithick] (0.5, 0) edge (1.5, 1.5) edge (1.5, 0.5) edge (1.5, -1.5)
                  (2.5, 1.5) edge (3.25, 1.5) (2.5, 0.5) edge (3.25, 0.5) (2.5, -1.5) edge (3.25, -1.5)
                  (6.25, 1) edge  node[align=center,above=-5pt,label] {Success\\decisions} (8, 1)
                  (6.25, -1) edge node[align=center,above=-5pt,label] {Failure\\decisions} (8, -1)
                  (10, 1) edge (11, 1)
                  (10, -1) edge (11, -1)
                  (13, 1) edge (14.5, 0)
                  (13, -1) edge (14.5, 0);
                  
        \path[semithick] (6.25, 1.5) edge (6.25, -1.5) (5.75, 1.5) edge (6.25, 1.5) (5.75, 0.5) edge (6.25, 0.5) (5.75, -1.5) edge (6.25, -1.5);

        \node[align=center] (c1) at (4.5,  1.5) {$\Succ{\Request_1}$ or $\Fail$};
        \node[align=center] (c2) at (4.5,  0.5) {$\Succ{\Request_2}$ or $\Fail$};
        \node[align=center]      at (4.5, -0.25) {$\vdots$};
        \node[align=center] (cm) at (4.5, -1.5) {$\Succ{\Request_\m}$ or $\Fail$};

        \node[align=center] (permute) at (9,   1) {order\\requests};
        \node[align=center] (execute) at (12,  1) {execute\\requests};
        \node[align=center] (select) at  (9,  -1) {select new\\primaries};
        \node[align=center] (replace) at (12, -1) {replace\\primaries};

        \node[dot] (eor) at (14.5, 0) {};
        \node[node_text,right] at (eor) {End};

        \draw[decoration={brace,amplitude=5pt},decorate,thick] (8, 1.7) -- (13, 1.7);
        \node[label,above] at (10.5, 1.8) {Deterministic round execution};
        
        \draw[decoration={brace,amplitude=5pt,mirror},decorate,thick] (8, -1.7) -- (13, -1.7);
        \node[label,below] at (10.5, -1.8) {Unified primary replacement};
    \end{tikzpicture}
    \caption{A high-level overview of a replica $\Replica$. The replica coordinates a single consensus round among $\m$ instances of some \emph{consensus protocol}. Each instance yields a consensus decision. The success decisions yield a set of client requests, which are executed in a deterministic order. The failure decisions are collected and can be used to replace primaries in a unified manner.}\label{fig:mbft_replica_overview}
\end{figure*}
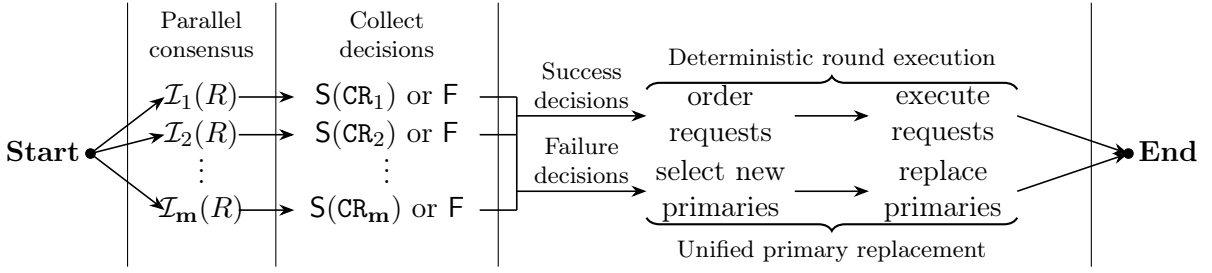

In each \emph{parallelized consensus round}, we first allow each of the $\m$ instances to independently reach a consensus decision. 
Next, we collect these decisions. The success decisions---of the form $\Succ{\Request}$---are executed in a \emph{deterministic fashion}. 
The failure decisions---of the form $\Fail$---are used to recover the instances involved in these decisions. 
To recover these instances, we replace their respective primaries, which we explain later in this section.

We use $\Dec{\rn}{i}$, $1 \leq i \leq \m$, to denote the consensus decision of instance $\Instance{i}$, 
agreed by all the non-faulty replicas in round $\rn$. Similarly, we use $\Primary{i,\rn}$ to indicate the primary of the $i$-th instance in round $\rn$. We write $\Decs{\rn} = \{ \Dec{\rn}{1}, \dots, \Dec{\rn}{\m} \}$ to represent the set of $\m$ consensus decisions agreed upon by all the non-faulty replicas. Finally, we write 
\[
    \SuccS{\Decs{\rn}} = \{ i \mapsto \Succ{\Request} \mid \Dec{\rn}{i} = \Succ{\Request}, 1 \leq i \leq \m \};\qquad
    \FailS{\Decs{\rn}} = \{ i \mapsto \Fail \mid \Dec{\rn}{i} = \Fail, 1 \leq i \leq \m \},
\]
to denote the partitioning of $\Decs{\rn}$ into sets of success decisions and failure decisions.

In Section~\ref{ss:det_rex}, we describe how to determine the order of execution of the client requests in $\SuccS{\Decs{\rn}}$ and in Section~\ref{ss:primfail}, we describe how to deal with primary failure in a coordinated manner in response to the failure decisions in $\FailS{\Decs{\rn}}$.  We discuss the assignment of clients to instances in the following section as part of the efforts to optimize parallelization benefits by removing the need for round-based step-wise operations.

\subsection{Deterministic round execution}\label{ss:det_rex}

The correctness of the underlying consensus protocol, used by instances $\Instance{1}, \dots, \Instance{\m}$, guarantees that each non-faulty replica derives the same set of client requests $\SuccS{\Decs{\rn}}$ in round $\rn$. Hence, non-faulty replicas only need to determine the order of execution of these client requests.

A simple solution would be to order the client requests based on their instance identifiers: first execute the client request of $\Instance{1}$ (if any), then execute the client request of $\Instance{2}$ (if any), and so on until all the requests are executed. Although this approach \emph{guarantees a unique sequential order} among all the executed client requests across all the non-faulty replicas, the approach also gives earlier instances disproportional control over execution. We illustrate this next:
\begin{example}
Consider a financial service in which client requests are of the form 
\[
    \operatorname{transfer}(A, B, n, m) \GETS \texttt{if $\operatorname{amount}(A) > n$ then $\operatorname{withdraw}(A, m)$; $\operatorname{deposit}(B, m)$}.
\]
Let $\Request_1 = \operatorname{transfer}(\text{Alice}, \text{Bob}, 500, 200)$ and $\Request_2 = \operatorname{transfer}(\text{Bob}, \text{Eve}, 400, 300)$ be client requests. Execution of $\Request_1$  influences the outcome of execution of $\Request_2$: if $200 \leq \operatorname{amount}(\text{Bob}) < 400$, then execution of $\Request_1$ before $\Request_2$ will result in a transfer of $300$ to $\text{Eve}$. If $\Request_1$ is executed after $\Request_2$, then $\text{Eve}$ will not receive anything. 
Hence, by choosing a predictable order of execution, earlier instances in the ordering can influence the execution of any requests accepted by later instances. 
\end{example}  

To resolve the illustrated shortcoming, we propose a method to deterministically select a different permutation of the order of execution in every round. Note that for any sequence $S$ of $k = \abs{S}$ values, there exist $k!$ distinct permutations. We write $P(S)$ to denote the set of permutations of $S$. As $\abs{P(S)} = k!$, there exists a bijection $f_S : \{ 0, \dots, k!-1 \} \rightarrow P(S)$. 
Next, we define the function $f_S$ recursively. We have the following:
\[
f_S(i) = 
\begin{cases}
    S &\text{if $\abs{S} = 1$};\\
   f_{S \difference S[q]}(r) \Concat S[q]   &\text{if $\abs{S} > 1$},
\end{cases}
\]
in which $q = i \div (\abs{S}-1)!$ is the quotient and $r = i \bmod (\abs{S}-1)!$ is the remainder of integer division by $(\abs{S}-1)!$.
\begin{lemma}\label{lem:bijection}
Function $f_S$ is a bijection from $\{ 0, \dots, \abs{S}! - 1 \}$ to all possible permutations of $S$.
\end{lemma}
\begin{proof}
The proof is by induction on the size of $S$. The base case is $\abs{S} = 1$, in which case $f_S = \{ 0 \mapsto S \}$, a bijection. As the induction hypothesis, we assume that $f_{S'}$ is a bijection for all $S'$ with $\abs{S'} \leq j$. Next, consider the case $f_{S}$ with $\abs{S} = j + 1$. Observe that, for each $s \in S$, there exist $j!$ permutations of $S$ that end with $s$. The computation of $q = i \div (\abs{S}-1)! = i \div j!$ and $r = i \bmod (\abs{S}-1)! = i \bmod j!$ divides all possible values $i$ into $(j+1)$ ranges of $j!$ values each. Hence, each $s \in S$ is chosen via $j!$ different values of $i$, and we have $j!$ different possible values for $r$ for each choice of $s$. The function $f_S$ chooses $s = S[q]$. By the induction hypothesis, the function $f_{S \difference s}(r)$ is a bijection. Thus, we conclude that the function $f_S$ is also a bijection.
\end{proof}

The result of $f_S(\rn \bmod \abs{S}!)$, on a sequence $S$ of all the consensus decisions of round $\rn$, can be used by all non-faulty replicas as a deterministic order of execution. This approach ensures that each instance receives equal opportunity to propose the client request to be executed first, across all the replicas. However, this approach is highly predictable.  Hence, as a further improvement, we can use the value $h = \Hash{R}$, with $R$ the set of client requests accepted in round $\rn$, instead of the round number $\rn$ to determine the order of execution. Assuming at least one primary is non-malicious ($\m > \f$), this value $h$ is only known after completion of the round $\rn$ as the malicious primaries cannot effectively collude to obtain a certain order of execution. Figure~\ref{fig:execute} presents the pseudo-code for the execution protocol.

\begin{figure}[t!]
    \begin{minipage}{0.6\textwidth}
        \begin{myprotocol}
                \PROTOCOL{Replica}{$\Replica \in \NonFaulty{\Service}$}
                \EVENT{collected $\SuccS{\Decs{\rn}}$ in round $\rn$}
                    \STATE Let $S \GETS [ \Request \mid (i \mapsto \Succ{\Request}) \in \SuccS{\Decs{\rn}} ]$, ordered on identifier $i$.
                    \FOR{$\Request \in f_{S}(\Hash{S} \bmod \abs{\SuccS{\Decs{\rn}}}!)$ with $\Request$ requested by $\Client \in \Clients$}
                        \STATE Execute $\Request$, yielding a resulting value $v$.
                        \STATE Send $v$ to $\Client$.
                    \ENDFOR
                \ENDEVENT
        \end{myprotocol}
    \end{minipage}
    \hfill
    \begin{minipage}{0.37\textwidth}
        \caption{The execution protocol running at each non-faulty replica of service $\Service = (\Clients, \Replicas, \Faulty)$.}\label{fig:execute}
    \end{minipage}
\end{figure}

\begin{proposition}\label{prop:execute}
The execution protocol of Figure~\ref{fig:execute} guarantees that every non-faulty replica executes client request in the same order. If $\m > \f$, then malicious replicas cannot control the order of execution.  If $\n > 2\f$, then clients will always receive at least $\f + 1$ identical results (from at least $\f+1$ non-faulty replicas) and can reliably detect successful execution of their request.
\end{proposition}

\subsection{Dealing with primary failure}\label{ss:primfail}

The easiest way to deal with primary failure is by shutting down the instance coordinated by that primary. 
This approach would work well for some time, as in most practical settings the set of faulty replicas is relatively stable.  However, for high availability, we need to consider a more dynamic setting in which we only know that at most $\f$ replicas are faulty in a specific window of time (as faulty replicas can recover and non-faulty replicas can become faulty). Hence, we propose two targeted methods to deal with primary failure.

\subsubsection{Discarding primaries and in-place recovery}\label{sss:pf_ignore}

An easy way to deal with faulty primaries is by permitting a certain delay for the failing instance to \emph{recover}, after which all the non-faulty replicas are instructed to transfer control back to the previously-failed primary.  This requires all the non-faulty replicas to agree on a delay and this delay needs to provide the faulty primary sufficient time to recover.  If the delay is too short, then it would result in repeated primary failure, which in turn would lead to multiple failed attempts to transfer control back to the primary, an unnecessary cost for all the replicas. To determine the right delay, the replicas can start with a small value (in number of rounds) and double this value after each failure. This approach does not necessitate any coordination between the replicas or between the instances.  Further, this approach even works when $\m > \n -\f$, in which case some primaries might always be faulty while no non-faulty replicas are available to replace them. 

\subsubsection{Unified primary replacement}\label{sss:pf_uni_prpl}

A second way to deal with faulty primaries is by replacement. Indeed, when $\m \leq \n - \f$ our paradigm can aim for a stable set of $\m$ non-faulty primaries by replacing a failed primary by another available replica. 
This approach is akin to the one taken by traditional \BFT{}-style consensus protocols. In such protocols, the failure of the primary $\Replica[p]$ is detected whenever the behavior of $\Replica[p]$ prevents successful consensus decisions. After detecting the failure of $\Replica[p]$, all the non-faulty replicas switch to the \emph{next primary}. This next primary is deterministically selected by choosing the replica following $\Replica[p]$, that is, by choosing the replica $\Replica$ with $\ID{\Replica} = \ID{\Replica[p]} + 1$.\footnote{Recently, a few consensus protocols proposed choosing primaries uniformly at random using a distributed random coin~\cite{rand}. Such an approach replaces the malicious primaries with a non-faulty replica with high probability. Our parallelization paradigm can easily be extended to use this approach.} 
However, such a selection strategy may not work with our paradigm as we require all the $\m$ instances to have distinct primaries. 
\begin{example}
Consider instances $\Instance{1}$ and $\Instance{2}$ with primaries $\Replica = \Primary{1}$ and $\Replica[q] = \Primary{2}$ with $\ID{\Replica} = 1$ and $\ID{\Replica[q]} = 2$. Consider the following two consensus decisions made during round $\rn$:
\begin{enumerate}
    \item $\Decs{\rn} = [\Fail, \Succ{\Request}]$. Here, instance $\Instance{1}$ needs to replace its primary. If $\Instance{1}$ chooses the replica following $\Replica$, replica $\Replica[q]$, then both instances end up with the same primary. 
    \item $\Decs{\rn} = [\Fail, \Fail]$. Here, both the instances need to replace their primaries.  If $\Instance{1}$ chooses the replica following $\Replica$, replica $\Replica[q]$, then it will end up with a known faulty primary. If $\Instance{1}$ realizes that $\Replica[q]$  was already in use by $\Instance{2}$,  then $\Instance{1}$ would choose the replica following $\Replica[q]$. Unfortunately, in such a case, $\Instance{2}$ would do the same and both instances  would end up with the same primary.
\end{enumerate}
These cases illustrate that primary replacement would fail without coordination between the instances. Hence, this coordination is an essential task in our parallelization  paradigm.
\end{example}

\begin{figure}[t!]
    \begin{minipage}{0.6\textwidth}
        \begin{myprotocol}
            \PROTOCOL{Replica}{$\Replica \in \NonFaulty{\Service}$}
            \STATE $\Var{failed} \GETS \emptyset$.
            \STATE $\Var{primary} \GETS \{ (i \mapsto \Replica_i) \mid 0 \leq i < \m, \ID{\Replica_i} = i \}$.\label{fig:uni_prot:init}
            \EVENT{collected $\FailS{\Decs{\rn}}$ in round $\rn$}
                \STATE $\Var{failed} \GETS \Var{failed} \union \{ \Primary{i} \mid (i \mapsto \Fail) \in \FailS{\Decs{\rn}} \}$.\label{fig:uni_prot:failed}
                \FOR{$1 \leq i \leq \m$ with $(i \mapsto \Fail) \in \FailS{\Decs{\rn}}$}\label{fig:uni_prot:replace_loop}
                    \STATE $\Var{Im} \GETS \{ \Var{primary}[i] \mid 1 \leq i \leq \m \}$.
                    \STATE Choose the replica $\Replica[p] \in (\Replicas \difference (\Var{failed} \union \Var{Im}))$ with smallest $\ID{\Replica[p]}$.\label{fig:uni_prot:choose_prim}
                    \STATE $\Var{primary}[i] \GETS \Replica[p]$.\label{fig:uni_prot:assign}
                    \STATE Inform instance $\Instance{i}$ of its new primary $\Replica[p]$.\label{fig:uni_prot:inform}
                \ENDFOR
            \ENDEVENT
        \end{myprotocol}
    \end{minipage}
    \hfill
    \begin{minipage}{0.37\textwidth}
        \caption{The unified primary replacement protocol running at each non-faulty replica of service $\Service = (\Clients, \Replicas, \Faulty)$.}\label{fig:uni_prot}
    \end{minipage}
\end{figure}

We introduce a \emph{unified primary replacement protocol}, which facilitates coordinated primary replacement among the instances. The protocol requires each non-faulty replica to maintain an internal state $(\Var{failed}, \Var{primary})$, where $\Var{failed} \subseteq \Faulty$ is the set of \emph{known faulty replicas}, and $\Var{primary} : \{ 1, \dots, \m \} \rightarrow (\Replicas \difference \Var{failed})$ is an injective function that maps each instance $\Instance{i}$, $1 \leq i \leq \m$, onto its primary, that is, $\Var{primary}(i) = \Primary{i}$. The function $\Var{primary}$ never maps to known faulty replicas. Figure~\ref{fig:uni_prot} presents the pseudo-code for this protocol. Formally, the unified primary replacement protocol maintains the following invariants:
\begin{invariant}\label{inv:uni_prot}
Let $\Service$ be a service. We write $\Var{failed}_{\rn}(\Replica)$ and $\Var{primary}_{\rn}(\Replica)$ to denote the value of these variables at non-faulty replica $\Replica \in \NonFaulty{\Service}$ at the start of round $\rn$. Further, we use $\Var{primary}_{\rn}(\Replica)[i]$, $1 \leq i \leq \m$, to denote the primary of instance $\Instance{i}$. For every $\Replica, \Replica[q] \in \NonFaulty{\Service}$, the following properties hold at the start of every round $\rn$:
\begin{enumerate}
    \item\label{inv:uni_prot:failed} $\Var{failed}_{\rn}(\Replica) = \{ \Primary{i,j} \mid (i \mapsto \Fail) \in \FailS{\Decs{j}}, 1 \leq j < \rn \}$;
    \item\label{inv:uni_prot:primary} $\Var{primary}_{\rn}(\Replica)$ is an injective function and $\Var{primary}_{\rn}(\Replica)[i] \in (\Replicas \difference \Var{failed})$, $1\leq i \leq \m$;
    \item\label{inv:uni_prot:agree} $\Var{primary}_{\rn}(\Replica) = \Var{primary}_{\rn}(\Replica[q])$ and $\Var{failed}_{\rn}(\Replica) = \Var{failed}_{\rn}(\Replica[q]) \subseteq \Faulty$.
\end{enumerate}
 \end{invariant}

\begin{proposition}\label{prop:unify}
The unified primary replacement protocol of Figure~\ref{fig:uni_prot} maintains Invariant~\ref{inv:uni_prot}.
\end{proposition}
\begin{proof}
Let $\Service = (\Clients, \Replicas, \Faulty)$ be a service and let $\Replica, \Replica[q] \in \NonFaulty{\Service}$. Initially, due to Line~\ref{fig:uni_prot:init}, we have $\Var{failed}_0(\Replica) = \Var{failed}_0(\Replica[q])  = \emptyset$ and $\Var{primary}_0(\Replica)[i] = \Var{primary}_0(\Replica[q])[i] = \Replica[p]_i$, $1 \leq i \leq \m$,  in which $\ID{\Replica[p]_i} = i$. Hence, Invariant~\ref{inv:uni_prot} initially holds.

Next, we assume that Invariant~\ref{inv:uni_prot} holds at the start of round $\rn$, and we prove that it again holds at the start of round $\rn+1$. At Line~\ref{fig:uni_prot:failed}, the failed primaries in $\FailS{\Decs{\rn}}$ are added to $\Var{failed}$. 
As the underlying consensus protocol ensures that every replica decides on the same set $\FailS{\Decs{\rn}}$, 
all the non-faulty replicas make identical changes to $\Var{failed}$. 
Hence, we are assured that Invariant~\subref{inv:uni_prot}{inv:uni_prot:failed} is maintained. 
The loop at Line~\ref{fig:uni_prot:replace_loop} will replace every newly detected faulty primary by a freshly chosen primary that is not yet in use (not in $\Var{Im}$) and is not a known faulty replica (not in $\Var{failed}$). As all non-faulty replicas agree on $\FailS{\Decs{\rn}}$, they process values in $\FailS{\Decs{\rn}}$ in the same order (Line~\ref{fig:uni_prot:replace_loop}). 
Further, each non-faulty replica choose new primaries deterministically (Line~\ref{fig:uni_prot:choose_prim}) and makes the same changes to $\Var{primary}$. 
Hence, we are assured that Invariants~\subref{inv:uni_prot}{inv:uni_prot:primary} and~\subref{inv:uni_prot}{inv:uni_prot:agree} are maintained. 
Finally, at Line~\ref{fig:uni_prot:inform} each instance $\Instance{i}$ that decided $\Fail$ in round $\rn$ gets assigned a new primary $\Var{primary}[i]$. 
Thus, we conclude that Invariant~\ref{inv:uni_prot} holds.
\end{proof}

For environments in which the set of faulty replicas is ever changing, the unified primary replacement protocol can easily be tweaked 
such that faulty primaries are eventually reconsidered 
(for example, by reintroducing the earliest failed replicas after all the other replicas have failed).

\section{Optimizing parallelization to increase performance}\label{sec:waitfree}

In the previous section, we presented a paradigm for parallelizing consensus protocols. To simplify the presentation, we presented a \emph{step-wise} design whose practical 
implementation would require substantial coordination between instances 
(for example, via the use of locks).  
In practice, this step-wise design incurs a lot of \emph{waiting}, which we illustrate next:

\begin{example}\label{ex:waiting}
Consider service $\Service$ working on a consensus round $\rn$ and let $\Replica \in \NonFaulty{\Service}$ be a non-faulty replica. We describe four cases in which the design presented in Section~\ref{sec:stepwise} induces waiting:
\begin{enumerate}
    \item\label{ex:waiting:execute} The execution of client requests consumes time and forces all the instances to wait until completion.
    \item\label{ex:waiting:var} In a practical setting, there can be a variation in message delivery time, which can lead to instances making consensus decisions at different speeds. For example, a temporary hiccup in the network can cause the instance $\Instance{1}$ to take twice the time it takes the other instances to make a consensus decision in round $\rn$. Hence, all the other instances would have to wait for $\Instance{1}$ to complete with a successful consensus decision.
    \item\label{ex:waiting:throttle} A faulty primary $\Primary{i}$ coordinating instance $\Instance{i}$, $1 \leq i \leq \m$, can actively throttle the speed at which its instance makes a consensus decision, which delays all the other instances. 
    \item\label{ex:waiting:crash} A primary $\Primary{i}$ coordinating instance $\Instance{i}$, $1 \leq i \leq \m$, can crash. Existing consensus protocols~\cite{pbft,zyzzyva} detect such a crash through large timeout values. 
This forces the other instances, whose consensus throughput should only be limited by the network latency, to wait for long idle periods for a primary failure to resolve.
\end{enumerate}
\end{example}

Waiting reduces the attainable performance (given the available resources). 
Fortunately, all the above described forms of waiting can be eliminated from our paradigm. First, in Section~\ref{ss:remove_wait}, we describe how to eliminate waiting. 
This complicates dealing with client requests, and, hence, in Section~\ref{ss:client_wait}, 
we describe these complications and provide solutions to resolve them.

\subsection{Making parallelization wait-free}\label{ss:remove_wait}

To ensure the correctness of our parallelization paradigm, we do not require any instance to wait for the other instances. 
Consider a service working on consensus round $\rn$. 
First, we observe that the \emph{execution of client requests} in round $\rn$ has no  influence on the consensus decisions of future rounds. 
Second, the instances arriving at successful consensus decisions do not require any coordination.
The only required coordination between the instances is the unified primary replacement (Section~\ref{sss:pf_uni_prpl}), 
which is limited to instances with failed primaries. 
Hence, instances that arrived at successful consensus decisions in the current round are free to make consensus decisions for the future rounds, 
while the \emph{execution} of the client requests of previous rounds occurs in the background, and 
the other instances are still making consensus 
decisions for the current round. 
Thus, dealing with Example~\ref{ex:waiting}, Case~\ref{ex:waiting:execute} and Case~\ref{ex:waiting:var}, is straightforward.

To arrive at a fully wait-free design, we must also address the malicious behaviors described in Example~\ref{ex:waiting}, Case~\ref{ex:waiting:throttle} and Case~\ref{ex:waiting:crash},
 and deal with any structural differences in speed of the instances. 
Note that if these behaviors are left unhandled, then they can cause unbounded delays between the acceptance and execution of a client request, which we illustrate next:

\begin{example}\label{ex:execute_delay}
Consider a service with $\m = 2$ instances, where instance $\Instance{1}$ makes a consensus decision every $10\si{\micro\second}$ and $\Instance{2}$ makes a consensus decision every $20\si{\micro\second}$. As $\Instance{2}$ operates slower, 
it determines the speed by which the system can complete a consensus round. 
Consequently, over time $\Instance{1}$ will make consensus decisions of an ever growing set of client requests that are awaiting execution. 
Further, the delay between $\Instance{1}$ accepting a client request and its execution grows without a bound.

Similarly, consider the case in which both  instances make consensus decisions every $10\si{\micro\second}$. Assume that in round $\rn$ the primary $\Primary{2}$ of instance $\Instance{2}$ fails. Further, assume it takes $500\si{\micro\second}$ to detect such a failure and $100\si{\micro\second}$ to replace the primary. In such a case, $\Instance{1}$ would have made $60$ consensus decisions before $\Instance{2}$ resumes normal operation. Hence, after round $\rn$ all the client requests accepted by $\Instance{1}$ will see an additional delay of $600\si{\micro\second}$.
\end{example}

The situations illustrated in Example~\ref{ex:execute_delay} are among the several cases in which the client requests accepted by well-performing instances would see their execution unnecessary delayed due to interference (delays) from other instances. Next, we show how our parallelization paradigm can address these cases using a simple yet effective principle:

\begin{definition}\label{def:soft_failure}
We say that an instance $\Instance{i}$, $1\leq i \leq \m$, suffers a \emph{soft failure} if it is working on a consensus decision in round $\rn$, while some other instance is already working on a consensus decision in round $\rn + \sigma$. We call $\sigma$ the \emph{gap size}, which is determined by the network latency and the timeout used to detect failures.

When an instance $\Instance{i}$, $1 \leq i\leq \m$, suffers any type of failure in round $\rn$, including soft failures, 
then $\Instance{i}$ is excluded from contributing to the next $\varepsilon$ consensus rounds. 
We call $\varepsilon$ the \emph{skip size}, which is determined by the gap size and time to replace the faulty primary. 
Hence, in case of the failure of instance $\Instance{i}$ in round $\rn$, $\Instance{i}$ starts making consensus decisions for consensus round $\rn + \varepsilon$.
\end{definition}

The soft failure principle is based on the assumption that all instances coordinated by the non-faulty primaries should reach successful consensus decisions roughly within the same time (as they all operate in the same environment). Hence, instances that lag by a significant margin $\sigma$ could be led by a faulty primary.\footnote{A similar assumption is the basis of \rBFT{}~\cite{rbft}, see Section~\ref{sec:related}.} 
Each replica locally detects the soft failure of its instance $\Instance{i}$, $1\leq i \leq \m$, and 
uses the fault detection infrastructure of the underlying consensus protocol to work towards ending the ongoing consensus round with a decision $\Fail$. 
The concept of soft failures addresses the execution delays due to underperforming instances, whereas the skipping of consensus rounds allows previously-failed instances to catch up with the other instances.

\begin{theorem}\label{thm:wait_free}
Instances are \emph{wait-free}: instances can make successful consensus decisions without outside interference and, at the same time, the delay between an instance accepting a client request and the replicas executing this client request is upper-bounded.
\end{theorem}

Note that if unified primary replacement (Section~\ref{sss:pf_uni_prpl}) is used to deal with primary failures,
then natural fluctuations in the performance of an instance can cause an unjust replacement of its primary. 
Hence, we allow treatment of soft failures as temporarily failures from which primaries can eventually recover: a replica that \emph{fails soft} can be considered non-faulty when the unified primary replacement protocol runs out of replicas it considers non-faulty.

\subsection{Consistent handling of client requests}\label{ss:client_wait}

Consensus protocols facilitate execution of client requests in a consistent manner across all the non-faulty replicas. Usually, consensus protocols can also aim at executing these client requests in the order they were sent by the clients. Maintaining this consistent ordering becomes harder when parallelizing consensus protocols:

\begin{example}
Consider a service $\Service$ with two instances $\Instance{1}$ and $\Instance{2}$ and client $\Client$. If both $\Instance{1}$ and $\Instance{2}$ are about to make a consensus decision for round $\rn$ and $\Client$ sends a request $\Request$ to both instances, then both instances might propose the same request, which would waste resources. If $\Client$ sends requests $\Request_1$ to $\Instance{1}$ and $\Request_2$ to $\Instance{2}$, then the order in which these requests are executed is subject to decisions made by the execution protocol (Section~\ref{ss:det_rex}). Moreover, due to the wait-free design, instances $\Instance{1}$ and $\Instance{2}$ can be in completely different rounds when proposing $\Request_1$ and $\Request_2$, again making the order of execution independent of the order in which $\Client$ requested $\Request_1$ and $\Request_2$.
\end{example}

If consistent ordering of execution of client requests is not necessary, e.g., if client requests operate on conflict-free replicated data types~\cite{crdt}, then clients can simply send their transactions to arbitrary instances. If consistent ordering is necessary, then the straightforward way to guarantee ordering is to assign each client to a unique instance. 

Let $\Service = (\Clients, \Replicas, \Faulty)$ be a service. If we assume $\abs{\Clients} > \m$, then we can assign clients to instances in a round-robin manner by requiring that the instance $\Instance{i}$, $1 \leq i \leq \m$, only deals with client requests of clients $\Client \in \Clients$ with $i = \ID{\Client} \bmod \m$. We notice that a client $\Client$ can be assigned to an instance with a faulty primary that might ignore the client request. In existing consensus protocols, this behavior  eventually leads to the primary being detected as faulty. 
Hence, based on how we deal with the faulty primaries there are two ways to guarantee service for $\Client$.
\begin{enumerate}
\item If faulty primaries are replaced (Section~\ref{sss:pf_uni_prpl}), then unified primary replacement assures that eventually 
a non-faulty primary will coordinate the instance and propose the requests of $\Client$.
\item If faulty primaries are not replaced (Section~\ref{sss:pf_ignore}), then the requests of client $\Client$ could get indefinitely ignored by a faulty primary that never recovers. In such a case, we allow $\Client$ to switch to another instance $\Instance{i}$, $1\leq i \leq \m$. To do so, $\Client$ sends an instance-change request to $\Instance{i}$. If $\Instance{i}$ gets this request, then it adds the change-request to reassign $\Client$ to $\Instance{i}$ to the consensus decision it is going to make in $\rn$ by proposing to all replicas to reassign the client to $\Instance{i}$ in round $\rn + 2\sigma$. After this proposal is requested, instance $\Instance{i}$ is able to propose client requests for $\Client$ after round $\rn + 2\sigma$ is executed. To assure balanced load among instances, a non-faulty instance only has to accept an instance-change request if it does not yet have $\lceil \abs{\Clients} / \abs{\NonFaulty{\Service}} \rceil$ clients assigned.
\end{enumerate} 

\subsection{Overview of wait-free designs}

We have explored several different designs for the parallelization of consensus protocols, each resulting in a valid and highly parallelized consensus protocol. Next, we summarize our findings.

\begin{theorem}\label{thm:main}
The parallelization paradigm we propose can turn a general consensus protocol into a \emph{high-performance parallelized wait-free consensus protocol} in which every client will eventually see its requests executed, the load of non-faulty replicas is evenly distributed, and the impact of faulty replicas is minimized. Additionally, the parallelization paradigm we propose can also turn partial consensus protocols into \emph{high-performance parallelized wait-free consensus protocols}.
\end{theorem}
\begin{proof}
We consider the following instances of the parallelization paradigm based on how one deals with primary failure and clients:
\begin{enumerate}
\item failed primaries are not replaced and clients can request instance reassignment; and
\item primaries are replaced and clients are assigned statically to instances.
\end{enumerate}

In both cases, \emph{non-divergence} of the parallelized protocol follows from non-divergence of the underlying consensus protocol, which assures that all replicas reach the same consensus decision for each instance in each round, and the deterministic round execution of the \emph{execution protocol} (Proposition~\ref{prop:execute}). Next, we look at \emph{termination}, which follows directly from termination of the underlying consensus protocol.

The equal distribution of load among all replicas follows from the uniform assignment of clients to each instance (when primaries are replaced) and by the upper-bound on the assigned clients when client can request instance reassignment. The impact of faulty replicas is minimized due to the resulting protocol being wait-free (Theorem~\ref{thm:wait_free}).
\end{proof}

\section{Related work}\label{sec:related}

There is an abundant literature on consensus protocols and, in specific, primary-backup consensus protocols (e.g.,~\cite{scaling,untangle,distalgo,distsys,byza,encybd,wild}). 
In this paper, we primarily focus on works that address the limitations of primary-backup protocols, as described in the Introduction. Several different approaches towards resolving some of these limitations have been considered in the literature.

\paragraph{Leader-free protocols.} Several leader-free protocols have been proposed to eliminate any issues arising from the discrepancy in responsibilities between primaries and backups, especially with respect to malicious primary behavior~\cite{honey,leaderfree}. In these leader-free designs, all replicas have the same responsibilities, the same load, and have the same impact on the system performance. Unfortunately, these leader-free designs come at high communication costs, making their practical usage limited. Recently, HoneyBadgerBFT proposed a leader-free design based on expensive asynchronous broadcast protocols and reduce the amortized per-request communication cost by making the batch size a function of the communication complexity of the protocol. We believe that simpler and more efficient primary-backup designs are more suitable for high-performance applications, and we view fine-tuning their performance (e.g., via parallelization as explored in this paper, or by applying amortized optimizations such as explored in HoneyBadgerBFT) as a more promising avenue for further development.

The Proof-of-Work (PoW) protocol and other similar protocols employed by cryptocurrencies~\cite{bitcoin,ethereum,encybd,untangle} are also leader-less. These protocols can be employed in permissionless environments in which participants can join and leave at any time~\cite{bitp2p}. Unfortunately, for many practical applications the computational costs of PoW are too high and the throughput too low~\cite{badcoin,badbadcoin,hypereal}. As permissioned \BFT{}-style protocols outperform PoW by several orders of magnitudes, this rules out their usage in the permissioned setting we study in this paper.

\paragraph{Redistribution of tasks.}  Recently, several complex consensus protocols have been proposed that use cryptographic techniques to reduce the communication costs of \BFT{}-style consensus~\cite{linbft,sbft}, which especially adds burden on the primary. In LinBFT, this is partially addressed by deferring some of the primary tasks to other replicas. Unfortunately, this complicates the design significantly, as the protocol not only has to detect and replace faulty primaries, but also detect and compensate for deferred faulty replicas. Moreover such deferral techniques are highly protocol-specific and only address issues related to the \emph{primary load}, not to the costs of \emph{primary replacement} or other \emph{malicious behavior} by the primary. Other designs, such as FastBFT~\cite{fastbft}, uses tree-based overlay networks and efficient message aggregation to reduce the total communication cost of the protocol. The usage of overlay networks is orthogonal to our approach and can reduce communication of the consensus protocol on which our paradigm relies.

\paragraph{Reducing malicious behavior.} Several works have observed that traditional \BFT{}-style consensus protocols only address a narrow set of malicious behavior, namely behavior that prevents any progress~\cite{rbft,aardvark,prime,spin}. Hence, several designs have been proposed to also address behavior that impedes performance without completely preventing progress. One such design is \rBFT{}, which uses parallelization not to improve performance---as we propose---but only to detect malicious behavior. In practice, the design of \rBFT{} results in poor performance at high costs. Another design is Spinning~\cite{spin}, which proposes to replace the primary every round. This would not incur the costs of \rBFT{}, while still reducing the impact of faulty replicas to severely reduce throughput. In our parallelization paradigm, however, we provide \emph{wait-free} consensus to instances with non-faulty primaries, which allows those instances to always process client requests at maximum throughput.

\section{Conclusions and future work}\label{sec:conclude}

In this paper, we propose a novel paradigm for parallelizing consensus protocols in a wait-free manner, 
thereby improving the system throughput by reducing the load on individual replicas and sharply reducing the impact of faulty replicas. Our techniques are protocol-agnostic, adjustable to several settings, and can be combined with many readily available primary-backup consensus protocols. 
Hence, our paradigm opens the door for the development of new and highly performant permissioned blockchain applications.

\bibliographystyle{plainurl}
\bibliography{sources}
\appendix

\end{document}